\documentclass[12pt]{article}
\usepackage[english]{babel}
\usepackage{amsmath}
\usepackage{amsthm}
\usepackage{amssymb}
\usepackage[dvips]{color}
\newtheorem{theorem}{Theorem}[section]

\newtheorem{proposition}[theorem]{Proposition}

\newtheorem{remark}[theorem]{Remark}

\newcommand{\cc}{\mathbb{C}}

\newcommand{\R}{\mathbb{R}}
\newcommand{\C}{\mathbb{C}}

\newcommand{\nn}{\nonumber}
\newcommand{\Z}{\mathbb{Z}}

\providecommand{\norm}[1]{\left\lVert#1\right\rVert}

\title{\bf Quantum harmonic oscillator  with superoscillating initial datum}

\author{R.V. Buniy\footnote{Schmid College of Science and Technology, Chapman University, Orange, CA 92866, USA},\ F. Colombo\footnote{Politecnico di
Milano, Dipartimento di Matematica, Via E. Bonardi 9, 20133 Milano, Italy},
\ I. Sabadini$^\dagger$,\ D.C. Struppa$^*$}

\begin{document}
\maketitle

\begin{abstract}
In this paper we study the evolution of superoscillating initial data for the quantum driven harmonic oscillator.
Our main result shows that superoscillations are amplified by the harmonic potential and that the analytic solution develops a singularity in finite time.
We also show that for a large class of solutions of the Schr\"odinger equation, superoscillating behavior at any given time implies superoscillating behavior at any other time.
\end{abstract}

\section{Introduction}

Superoscillating phenomena have been known in physics for a long time~\cite{aav,abook,av,berry2,b1,b4} and recently they have also been studied from the mathematical point of view~\cite{acsst4,acsst1,acsst2,acsst3,acsstYa80,acsst6,acsst7,acsst5}.
An interesting question concerning superoscillating functions in quantum mechanics is to ask whether the superoscillatory behavior persists when these functions evolve in time according to the Schr\"{o}dinger equation.
The prototype of a sequence of superoscillating functions that is of interest for this problem is $\{F_n(x)\}_{n=1}^\infty$, where
\begin{align*}
  &F_n(x)=\left[\cos{\left(\frac{px}{n\hbar}\right)}+ia\sin{\left(\frac{px}{n\hbar}\right)}\right]^n=\sum_{k=0}^n C_k(n,a)\exp{\left[\frac{ipx}{\hbar}\left(1-\frac{2k}{n}\right)\right]},\\
  &C_k(n,a)={n\choose k}\left(\frac{1+a}{2}\right)^{n-k}\left(\frac{1-a}{2}\right)^k,
\end{align*}
${n\choose k}$ is the binomial coefficient, $p\in\mathbb R$, $p\not=0$ and $a>1$.
Taking the limit $n\to\infty$, we find that $F_n(x)\to e^{iapx/\hbar}$ pointwise on $\mathbb{R}$ but uniformly only on compact sets in $\mathbb{R}$.
The lack of uniform convergence is responsible for many subtle properties of superoscillations.

It has been shown that the Cauchy problem (the initial value problem for the time-dependent Schr\"{o}dinger equation) for a quantum-mechanical free particle in one spacial dimension,
\begin{align*}
  \left[i\hbar\frac{\partial}{\partial t}+\frac{\hbar^2}{2m}\frac{\partial^2}{\partial x^2}\right]\psi_n(t,x)=0,\quad \psi_n(0,x)=F_n(x),
\end{align*}
has the solution
\begin{align*}
  \psi_n(t,x)=\sum_{k=0}^n C_k(n,a)\exp{\left[\frac{ipx}{\hbar}\left(1-\frac{2k}{n}\right)-\frac{ip^2 t}{2m\hbar}\left(1-\frac{2k}{n}\right)^2\right]}
  \label{}
\end{align*}
for every $(t,x)\in\R\times\R$, and moreover
\begin{align*}
  \psi(t,x)=\lim_{n\to\infty} \psi_n(t,x)=\exp{\left(\frac{iapx}{\hbar}-\frac{ia^2 p^2 t}{2m\hbar}\right)},
\end{align*}
where the convergence is uniform on compact sets.
This means that the superoscillatory phenomenon for the functions $\{\psi_n(t,x)\}_{n=1}^\infty$ persists for all times $t\in [-T,T]$ for any $T>0$.
A crucial fact in the proof of this result (see \cite{acsst3}) is that $\psi_n(t,x)$
can be written as
\begin{align*}
  &\psi_n(t,x)=U\left(t,\frac{d}{dx}\right)F_n(x),\\
  &U\left(t,\frac{d}{dx}\right):=\sum_{l=0}^\infty\frac{1}{l!}\left(\frac{i\hbar t}{2m}\right)^l \left(\dfrac{d}{dx}\right)^{2l},
\end{align*}
and with suitable mathematical techniques~\cite{bs,ehrenpreis,taylor} we can show that the operator $U\left(t,\frac{d}{dz}\right)$ acts continuously on a class of holomorphic functions with certain growth conditions that contains the analytic extension $\{F_n(z)\}$ of the sequence $\{F_n(x)\}$.
Thus, by restricting both to $\mathbb R$, we have
\begin{align*}
  \lim_{n\to\infty}U\left(t,\frac{d}{dx}\right)F_n(x)=U\left(t,\frac{d}{dx}\right)\lim_{n\to\infty}F_n(x)\nn\\ =U\left(t,\frac{d}{dx}\right)e^{iapx/\hbar}=\exp{\left(\frac{iapx}{\hbar}-\frac{ia^2 p^2 t}{2m\hbar}\right)}.
\end{align*}

The recently developed~\cite{acsst6, acsst7} mathematical strategy to generate superoscillatory functions consists in explicitly solving the Cauchy problem for suitable convolution equations with superoscillating initial datum.
This works well in the case when explicit solutions for Green's functions and propagators are known.
Here we solve the Cauchy problem for the quantum driven harmonic oscillator and find that the superoscillations are amplified by the potential and that the analytic solution develops a singularity in finite time.
Moreover, even for $a\in (0,1)$ the harmonic oscillator displays a superoscillatory phenomenon since the solution contains the term that increases arbitrarily with $t$ (up to the time of singularity).
This phenomenon does not occur for the free particle.

The physical interest in the quantum driven harmonic oscillator lies in its connection to the semiclassical (WKB) approximation in the path integral formulation of quantum mechanics~\cite{feynman,fh,schulman}.
It turns out that the leading  behavior of the propagator for any quantum-mechanical system when $\hbar\to 0$ is obtained by a series expansion of the corresponding action functional, which, without loss of generality, is given by the action for the driven harmonic oscillator.
Thus our results about the superoscillatory behavior of such an oscillator directly apply to a much broader problem of establishing presence of superoscillations in any quantum system (at least in the semiclassical limit).

\section{Driven harmonic oscillator}

A nonrelativistic point particle at $(t,x)\in\R\times\R^d$ with the Hamiltonian $H(t,x)$ and the wave function $\psi(t,x)$ is governed by the time-dependent Schr\"{o}dinger equation
\begin{align}
  \left[i\hbar\frac{\partial}{\partial t}-H(t,x)\right]\psi(t,x)=0.
  \label{schrodinger_eq}
\end{align}
The Green's function $G(t,x,t',x')$ for \eqref{schrodinger_eq} satisfies
\begin{align*}
  \left[i\hbar\frac{\partial}{\partial t}-H(t,x)\right]G(t,x,t',x')=i\hbar\delta(t-t')\delta(x-x')\label{green_function}
\end{align*}
and it follows that
\begin{align*}
  G(t,x,t',x')=\theta(t-t')\tilde{G}(t,x,t',x'),
\end{align*}
where $\theta$ is the Heaviside function and the propagator $\tilde{G}(t,x,t',x')$ for \eqref{schrodinger_eq} satisfies
\begin{align*}
  \left[i\hbar\frac{\partial}{\partial t}-H(t,x)\right]\tilde{G}(t,x,t',x')=0,\quad \tilde{G}(t,x,t,x')=\delta(x-x').
\end{align*}
As a result, the propagator is the kernel for the solution of the Cauchy problem since
\begin{align*}
  \psi(t,x)=\int_{\R^d}\tilde{G}(t,x,t',x')\psi(t',x')dx'
\end{align*}
for any $t\ge t'$.

Explicit solutions for Green's functions and propagators are known only for a few physical systems.
One of them is described by the Hamiltonian
\begin{align*}
    H(t,x)=-\frac{\hbar^2}{2m}\nabla_x^2+\frac{1}{2}m\omega^2(t)\norm{x}^2-f(t)\cdot x,
\end{align*}
which represents a $d$-dimensional harmonic oscillator of mass $m$  and time-dependent frequency $\omega(t)$ under the influence of the external time-dependent force $f(t)$.
The corresponding propagator~\cite{feynman,fh,montroll,gy,schulman} is given by the following proposition.
\begin{proposition}\label{proposition_harmonic_oscillator}
The propagator for the Hamiltonain
\begin{align*}
    H(t,x)=-\frac{\hbar^2}{2m}\nabla_x^2+\frac{1}{2}m\omega^2(t)\norm{x}^2-f(t)\cdot x
\end{align*}
is
\begin{align*}
  \tilde{G}(t,x,t',x')=\left[\frac{m}{2\pi i\hbar g(t,t')}\right]^{d/2}\exp{\left[\frac{i}{\hbar}S(t,x,t',x')\right]},
\end{align*}
where
\begin{align*}
  S(t,x,t',x')=\int_{t'}^t\left[\frac{1}{2}m\norm{\frac{dy(s)}{ds}}^2-\frac{1}{2}m\omega^2(s) \norm{y(s)}^2+f(s)\cdot y(s)\right]ds,
\end{align*}
$y(s)$ is the solution of the boundary value problem
\begin{align*}
  \frac{d^2 y(s)}{ds^2}+\omega^2(s) y(s)=f(s),\quad y(t)=x,\quad y(t')=x',
\end{align*}
and $g(t,t')$ is the solution of the initial value problem
\begin{align*}
  \frac{\partial^2 g(t,t')}{\partial t^2}+\omega^2(t) g(t,t')=0,\quad g(t',t')=0,\quad \left.\frac{\partial g(t,t')}{\partial t}\right\rvert_{t=t'}=1.
\end{align*}

\end{proposition}
The following four limiting cases of the above Hamiltonian and the corresponding propagator are of particular interest in physics:
\begin{itemize}
\item[(1)]
a free particle:
\[
\begin{split}
  &H(t,x)=-\frac{\hbar^2}{2m}\nabla_x^2,\\
  &g(t,t')=t-t',\\
  &S(t,x,t',x')=\frac{m\norm{x-x'}^2}{2(t-t')};
  \label{}
\end{split}
\]
\item[(2)]
a particle in a uniform field:
\[ \begin{split}
  &H(t,x)=-\frac{\hbar^2}{2m}\nabla_x^2-f\cdot x, \quad f=\textrm{const},\\
  &g(t,t')=t-t',\\
  &S(t,x,t',x')=\frac{m\norm{x-x'}^2}{2(t-t')}+\frac{1}{2}f(t-t')\cdot (x+x')-\frac{f^2(t-t')^3}{24m};
  \label{}
\end{split}  \]
\item[(3)]
a harmonic oscillator:
\[ \begin{split}
    &H(t,x)=-\frac{\hbar^2}{2m}\nabla_x^2+\frac{1}{2}m\omega^2 \norm{x}^2, \quad \omega=\textrm{const},\\
  &g(t,t')=\frac{\sin{\omega(t-t')}}{\omega},\\
  &S(t,x,t',x')=\frac{m\omega}{2\sin{\omega(t-t')}}\left[(\norm{x}^2+\norm{x}^{\prime 2})\cos{\omega(t-t')}-2xx'\right];
  \label{}
\end{split}  \]
\item[(4)]
a driven harmonic oscillator:
\[ \begin{split}
    &H(t,x)=-\frac{\hbar^2}{2m}\nabla_x^2+\frac{1}{2}m\omega^2 \norm{x}^2+f(t)\cdot x, \quad \omega=\textrm{const},\\
  &g(t,t')=\frac{\sin{\omega(t-t')}}{\omega},\\
  &S(t,x,t',x')=\frac{m\omega}{2\sin{\omega(t-t')}}\biggl[(\norm{x}^2+\norm{x}^{\prime 2})\cos{\omega(t-t')}-2x\cdot x'\nn\\
    &+2x\cdot I(t,t')+2x'\cdot I(t',t)-2J(t,t')\biggr],\\
  &I(t,t')=\frac{1}{m\omega}\int_{t'}^t f(s)\sin{\omega(s-t')}ds,\\
  &J(t,t')=\frac{1}{m^2\omega^2}\int_{t'}^t\int_{t'}^s f(s)\cdot f(s')\sin{\omega(t-s)}\sin{\omega(s'-t')}ds'ds.
  \label{}
\end{split}
\]
\end{itemize}

We proceed with the fourth example since the first three examples can be obtained as its appropriate limits.
We first prove a preliminary result that will be useful in the sequel.
\begin{proposition}\label{PropPsiA}
The solution of the Cauchy problem
\begin{align*}
  \left[i\hbar\frac{\partial}{\partial t}+\frac{\hbar^2}{2m}\nabla_x^2-\frac{1}{2}m\omega^2\norm{x}^2+f(t)\cdot x\right]\psi(t,x)=0, \quad \psi(0,x)=e^{iap\cdot x/\hbar}
  \label{cauchy_harmonic}
\end{align*}
is
\[
\begin{split}
  &\psi(t,x)=(\cos{\omega t})^{-d/2}\exp{\biggl\{\frac{im\omega}{2\hbar\sin{\omega t}}\biggl[-\frac{1}{\cos{\omega t}}\norm{x-p\frac{a\sin{\omega t}}{m\omega}-I(0,t)}^2}\nn
  \\
  &{+\norm{x}^2\cos{\omega t}+2x\cdot I(t,0)-2J(t,0)\biggr]\biggr\}}. \label{psiNoscA}\nn
  \end{split}
\]
\end{proposition}
\begin{proof}
  The proof follows by using the propagator, completing the square, and integrating over $x'$ with the help of the regularized integral
  \begin{align*}
    \int_{\R^d}e^{i\alpha \norm{x'}^2}\, dx'=\lim_{\beta\to 0^+}\int_{\R^d}e^{-(\beta-i\alpha)\norm{x'}^2}\, dx'=\Big(\frac{i\pi }{\alpha}\Big)^{d/2}.
    \label{}
  \end{align*}
\end{proof}
As a result we have the following theorem.
\begin{theorem}\label{TH23}
The solution of the Cauchy problem
\begin{align}
  &\left[i\hbar\frac{\partial}{\partial t}+\frac{\hbar^2}{2m}\nabla_x^2-\frac{1}{2}m\omega^2\norm{x}^2+f(t)\cdot x\right]\psi_n(t,x)=0,\label{cauchy_equation}\\
  &\psi_n(0,x)=\sum_{k=0}^n C_k(n,a)\exp{\left[\frac{ip\cdot x}{\hbar}\left(1-\frac{2k}{n}\right)\right]}
  \label{cauchy_initial}
\end{align}
is
\begin{align}
  \psi_n(t,x)&=(\cos{\omega t})^{-d/2}\exp{\biggl\{\frac{im\omega}{2\hbar\sin{\omega t}\cos{\omega t}}\Bigl[-\norm{x}^2\sin^2{\omega t}+2x\cdot I(t,0)\cos{\omega t}}\nn
  \\
  &{-2J(t,0)\cos{\omega t}+2x\cdot I(0,t)-\norm{I(0,t)}^2\Bigr]\biggr\}}\nn
  \\
  &
  \times \sum_{l=0}^\infty\frac{1}{l!}\left(\frac{i\hbar}{2m\omega}\sin{\omega t}\cos{\omega t}\right)^l(\nabla_x^2)^l\psi_n\left(0,\frac{x-I(0,t)}{\cos{\omega t}}\right).
  \label{cauchy_solution}
\end{align}
Moreover, if we set $\psi(t,x)=\lim_{n\to \infty}\psi_n(t,x)$, then
\begin{align}
    &\psi(t,x)=(\cos{\omega t})^{-d/2}\exp{\biggl\{\frac{im\omega}{2\hbar\sin{\omega t}}\biggl[-\frac{1}{\cos{\omega t}}\norm{x-p\frac{a\sin{\omega t}}{m\omega}-I(0,t)}^2}\nn\\ &{+\norm{x}^2\cos{\omega t}+2x\cdot I(t,0)-2J(t,0)\biggr]\biggr\}}.\label{cauchy_solution_limit}
\end{align}

\end{theorem}
\begin{proof} To prove that the solution of the Cauchy problem \eqref{cauchy_equation} and \eqref{cauchy_initial} is given by \eqref{cauchy_solution}, we
  observe that the initial datum $\psi_n(0,x)$ is a linear combination of the exponentials $e^{iap\cdot x/\hbar}$ with various values of $a$, and so by Proposition \ref{PropPsiA} we get \eqref{cauchy_solution}. To compute  $\lim_{n\to \infty}\psi_n(t,x)$ we proceed by steps.
\\
\\
Step 1. We observe that \eqref{cauchy_solution} can be written as
\begin{align*}
  \psi_n(t,x)&=(\cos{\omega t})^{-d/2}\exp{\biggl\{\frac{im\omega}{2\hbar\sin{\omega t}\cos{\omega t}}\Bigl[-\norm{x}^2\sin^2{\omega t}+2x\cdot I(t,0)\cos{\omega t}}\nn
  \\
  &{-2J(t,0)\cos{\omega t}+2x\cdot I(0,t)-\norm{I(0,t)}^2\Bigr]\biggr\}}\nn
  \\
  &
  \times \sum_{k=0}^n C_k(n,a)\exp{\left[\frac{i(1-2k/n)p\cdot(x-I(0,t))}{\hbar\cos{\omega t}}-\frac{i(1-2k/n)^2\norm{p}^2\tan{\omega t}}{2\hbar m\omega}\right]},
  \label{}
\end{align*}
which after using
\begin{align*}
  \exp{\left[\frac{i(1-2k/n)p\cdot(x-I(0,t))}{\hbar\cos{\omega t}}-\frac{i(1-2k/n)^2\norm{p}^2\tan{\omega t}}{2\hbar m\omega}\right]}\nn\\ =\exp{\left[\frac{i(1-2k/n)p\cdot(x-I(0,t))}{\hbar\cos{\omega t}}\right]}\sum_{l=0}^\infty\frac{1}{l!}\left[-\frac{i(1-2k/n)^2\norm{p}^2\tan{\omega t}}{2\hbar m\omega}\right]^l\nn\\ =\sum_{l=0}^\infty\frac{1}{l!}\left(\frac{i\hbar}{2m\omega}\sin{\omega t}\cos{\omega t}\right)^l(\nabla_x^2)^l\exp{\left[\frac{i(1-2k/n)p\cdot(x-I(0,t))}{\hbar\cos{\omega t}}\right]}
  \label{}
\end{align*}
gives \eqref{cauchy_solution}.

To proceed with the limit $n\to\infty$, we need to study the continuity of the operator
\begin{align}
  U(t,\nabla_x)=\sum_{l=0}^\infty\frac{1}{l!}\left(\frac{i\hbar}{2m\omega}\sin{\omega t}\cos{\omega t}\right)^l(\nabla_x^2)^l.
  \label{operator}
\end{align}
If $U$ is continuous on a function space that contains the functions $\{\psi_n\}_{n=1}^\infty$, we can write
\begin{align*}
  \psi(t,x)=\lim_{n\to\infty}\psi_n(t,x)=
  (\cos{\omega t})^{-d/2}\exp{\biggl\{\frac{im\omega}{2\hbar\sin{\omega t}\cos{\omega t}}\Bigl[-\norm{x}^2\sin^2{\omega t}}\nn\\ {+2x\cdot I(t,0)\cos{\omega t}}
    {-2J(t,0)\cos{\omega t}+2x\cdot I(0,t)-\norm{I(0,t)}^2\Bigr]\biggr\}}\nn\\
    \times U(t,x)\lim_{n\to\infty}\psi_n\left(0,\frac{x-I(0,t)}{\cos{\omega t}}\right)\nn\\ =
  (\cos{\omega t})^{-d/2}\exp{\biggl\{\frac{im\omega}{2\hbar\sin{\omega t}\cos{\omega t}}\Bigl[-\norm{x}^2\sin^2{\omega t}+2x\cdot I(t,0)\cos{\omega t}}\nn\\
    {-2J(t,0)\cos{\omega t}+2x\cdot I(0,t)-\norm{I(0,t)}^2\Bigr]\biggr\}} U(t,x)\psi\left(0,\frac{x-I(0,t)}{\cos{\omega t}}\right)
  \label{}
\end{align*}
since, from \cite{acsst1}, $\psi_n(0,(x-I(0,t))/\cos{\omega t})$ converges uniformly to $\psi(0,(x-I(0,t))/\cos{\omega t})$ for any $x\in[-M_1,M_1]\times\dotsb\times[-M_d,M_d]$, with $M_1,\dotsc,M_d$ any positive real numbers, for every fixed $t$ in $[0,\pi/(2\omega))$.
\\
\\
Step 2.
Now we show that $U$ is continuous on a function space that contains the functions $\{\psi_n\}_{n=1}^\infty$.
\\
We recall that if $X$ is an Analytically Uniform space (AU-space) \cite{bs,ehrenpreis,taylor},
then convolutors on $X$ can be defined in a standard way as follows. Let $\mathcal{F} X'$ be the space of the Fourier (or Fourier-Borel) transform of the elements of the dual of $X$. The definition of AU-space implies that $\mathcal{F} X'$ is a space of entire functions
which satisfy suitable growth conditions.
\\
Let $g$ be an entire function which, by multiplication, defines a continuous map from $\mathcal{F} X'$ to itself. Then, a convolutor on $X$ is the continuous operator on $X$ defined as the adjoint of the map that associates to $\varphi\in X'$
the element $\mathcal{F}^{-1}(g(\mathcal{F}(\varphi)))$.
Within this framework, if $X=\mathcal{O}(\mathbb{C}^d)$ is the space of entire functions on $\C^d$, then $X'$ is the space of analytic functionals and $\mathcal{F} X'$ is (topologically isomorphic to) the space ${\rm Exp}(\cc^d)$ of entire functions with exponential growth; see \cite{taylor}.
We replace $\nabla_x$ by the variable $z\in\cc^d$ in the operator \eqref{operator},
whose symbol
\begin{align}
\hat{U}(t,z):=\sum_{l=0}^\infty\frac{1}{l!}\Big(\frac{i\hbar}{2m\omega}\sin{\omega t}\cos{\omega t}\Big)^l\norm{z}^{2l}=\exp{\left(\frac{i\hbar\norm{z}^2}{2m\omega}\sin{\omega t}\cos{\omega t}\right)}
  \label{}
\end{align}
is continuous thanks to a simple extension of Theorem 3.3 in \cite{acsst3}:
{\it
for any value of $t$, the operator $U(t,\nabla_x)$ acts continuously on the space
$$
A_{2,0}:=\left\{f\in \mathcal{O}(\mathbb{C}^d)\ | \ \forall \varepsilon >0 \  \exists \ A_\varepsilon>0 \ | \ |f(z)|\leq A_\varepsilon e^{\varepsilon\norm{z}^2}    \right\}
$$
of entire functions of order less or equal $2$ and of minimal type.}

The functions $\{\psi_n\}_{n=1}^\infty$ extend to entire functions of order less than or equal 1 and of finite (i.e. exponential) type, and this space is clearly contained in $A_{2,0}$.
 \\
 \\
 Step 3. As the operator ${U}(t,\nabla_x)$ is continuous, we can write
\begin{align*}
  \psi(t,x)=
  (\cos{\omega t})^{-d/2}\exp{\biggl\{\frac{im\omega}{2\hbar\sin{\omega t}\cos{\omega t}}\Bigl[-\norm{x}^2\sin^2{\omega t}+2x\cdot I(t,0)\cos{\omega t}}\nn\\
  {-2J(t,0)\cos{\omega t}+2x\cdot I(0,t)-\norm{I(0,t)}^2\Bigr]\biggr\}}\nn\\ \times\sum_{l=0}^\infty\frac{1}{l!}\left(\frac{i\hbar}{2m\omega}\sin{\omega t}\cos{\omega t}\right)^l(\nabla_x^2)^l\exp\left(\frac{iap\cdot(x-I(0,t))}{\hbar\cos{\omega t}}\right),
  \label{}
\end{align*}
which gives \eqref{cauchy_solution_limit}.
\end{proof}

The amplitude and frequency of $\psi(t,x)$ simultaneously diverge for $t=(2k+1)\pi/(2\omega)$, $k\in\Z$, while
the frequency of oscillations of $\psi(t,x)$ in $x$ increases with $\norm{x}$ without bound for any $a>1$.
This is a consequence of the Hamiltonian of the harmonic oscillator generating the time evolution of a wave function with infinite norm.

Slightly modifying its proof, we can further generalize Theorem~\ref{TH23} by using a simple extension of the result proved in \cite{acsst1,acsst6}.
\begin{proposition}\label{sequences}
Let $a>1$, $p\in\R^d$, $x\in\R^d$ and $C_k(n,a)={n\choose k}\left(\frac{1+a}{2}\right)^{n-k}\left(\frac{1-a}{2}\right)^k$.
Then we have the following facts:
\begin{itemize}
\item[(a)]
The sequence $F_n(x)=\sum_{k=0}^n C_k(n,a)\exp{\left[\frac{ip\cdot x}{\hbar}\left(1-\frac{2k}{n}\right)\right]}$ converges pointwise to $F(x)=e^{iap\cdot x/\hbar}$ for all $x\in\R^d$ and the convergence is  uniform on the compact sets in $\R^d$.
\item[(b)]
If $q$ is an even number, then the sequence $Y_n(x)=\sum_{k=0}^nC_k(n,a) \exp{\left[\frac{ip\cdot x}{\hbar}(-i)^q\left(1-\frac{2k}{n}\right)^q\right]}$ converges pointwise to $Y(x)=e^{ip\cdot x(-ia)^q/\hbar}$ for all $x\in\R^d$ and the convergence is uniform on the compact sets in $\R^d$.
\item[(c)] If $q$ is an odd number, then the sequence $ Z_n(x)=\sum_{k=0}^nC_k(n,a)\exp{\left[\frac{p\cdot x}{\hbar}(-i)^q\left(1-\frac{2k}{n}\right)^q\right]} $ converges pointwise to  $Z(x)=e^{p\cdot x(-ia)^q/\hbar}$ for all $x\in\R^d$ and the convergence is uniform on the compact sets in $\R^d$.
\end{itemize}
\end{proposition}

\begin{theorem}
The solution of the Cauchy problem
\begin{align*}
  &\left[i\hbar\frac{\partial}{\partial t}+\frac{\hbar^2}{2m}\nabla_x^2-\frac{1}{2}m\omega^2\norm{x}^2+f(t)\cdot x\right]\psi_n(t,x)=0,\quad \psi_n(0,x)=Y_n(x)
\end{align*}
is
\begin{align*}
  \psi_n(t,x)&=(\cos{\omega t})^{-d/2}\exp{\biggl\{\frac{im\omega}{2\hbar\sin{\omega t}\cos{\omega t}}\Bigl[-\norm{x}^2\sin^2{\omega t}+2x\cdot I(t,0)\cos{\omega t}}
  \\
  &{-2J(t,0)\cos{\omega t}+2x\cdot I(0,t)-\norm{I(0,t)}^2\Bigr]\biggr\}}
  \\
  &
  \times \sum_{l=0}^\infty\frac{1}{l!}\left(\frac{i\hbar}{2m\omega}\sin{\omega t}\cos{\omega t}\right)^l(\nabla_x^2)^l Y_n\left(\frac{x-I(0,t)}{\cos{\omega t}}\right).
\end{align*}
Moreover, if we set $\psi(t,x)=\lim_{n\to \infty}\psi_n(t,x)$, then
\begin{align*}
  \psi(t,x)=(\cos{\omega t})^{-d/2}\exp{\biggl\{\frac{im\omega}{2\hbar\sin{\omega t}\cos{\omega t}}\Bigl[-\norm{x}^2\sin^2{\omega t}+2x\cdot I(t,0)\cos{\omega t}}\\
    -2J(t,0)\cos{\omega t}+2x\cdot I(0,t)-\norm{I(0,t)}^2\Bigr]-\frac{i(-ia)^{2q}\norm{p}^2\tan{\omega t}}{2m\hbar\omega}+\frac{i(-ia)^q p\cdot x}{\hbar\cos{\omega t}}\biggr\}.
\end{align*}
\end{theorem}

\begin{remark}{\rm
A similar result holds also for the sequence $Z_n(x)$ given in item (c) in Proposition \ref{sequences}.}
\end{remark}

\section{Persistence of superoscillations}
In this section we introduce another large class of superoscillating functions and show that for this class the superoscillatory behavior persists in time, and in fact cannot arise if the initial datum is not superoscillating.

\begin{theorem}
  Consider a sequence of functions $\{\psi_{n_1,\dotsc,n_d}(t,x)\}_{n_1,\dotsc,n_d=1}^\infty$ of the form
\begin{align}
  \psi_{n_1,\dotsc,n_d}(t,x)=\sum_{k_1=-n_1}^{n_1}\dotsb\sum_{k_d=-n_d}^{n_d} c_{n_1,\dotsc,n_d,k_1,\dotsc,k_d}(t)\exp{\left(\sum_{j=1}^d\frac{ik_j p_j x_j}{n_j\hbar}\right)},
  \label{psi_n}
\end{align}
with some coefficients $c_{n_1,\dotsc,n_d,k_1,\dotsc,k_d}(t)$.
Suppose that each $\psi_{n_1,\dotsc,n_d}(t,x)$ satisfies the time-dependent Schr\"{o}dinger equation
\begin{align}
  \left[i\hbar\frac{\partial}{\partial t}+\frac{\hbar^2}{2m}\nabla_x^2-W_{n_1,\dotsc,n_d}(t,x)\right]\psi_{n_1,\dotsc,n_d}(t,x)=0
  \label{schrodinger}
\end{align}
for a certain potential energy $W_{n_1,\dotsc,n_d}(t,x)$. Then the sequence $\{\psi_{n_1,\dotsc,n_d}(t,x)\}_{n_1,\dotsc,n_d=1}^\infty$is superoscillatory at any given time if and only if it is superoscillatory at any other time.
  \label{}
\end{theorem}
\begin{proof}
Substituting \eqref{psi_n} into \eqref{schrodinger}, we find
\begin{align}
  \sum_{k=-n}^n\biggl[i\hbar\frac{\partial c_{n_1,\dotsc,n_d,k_1,\dotsc,k_d}(t)}{\partial t}-\frac{1}{2m}\sum_{j=1}^d\frac{k_j^2 p_j^2}{n_j^2}c_{n_1,\dotsc,n_d,k_1,\dotsc,k_d}(t)\nn\\-W_{n_1,\dotsc,n_d}(t,x)c_{n_1,\dotsc,n_d,k_1,\dotsc,k_d}(t)\biggr]\exp{\left(\sum_{j=1}^d\frac{ik_j p_j x_j}{n_j\hbar}\right)}=0.
  \label{psi_n_expansion}
\end{align}
Although the set of functions $\bigl\{\exp{\bigl[\sum_{j=1}^d ik_j p_j x_j/(n_j\hbar)\bigr]}\bigr\}$ with $-n_j\le k_j\le n_j$, $1\le j\le d$ is not complete on
\begin{align*}
  \Omega=[-n\pi\hbar/p_1,n\pi\hbar/p_1]\times\dotsb\times[-n\pi\hbar/p_d,n\pi\hbar/p_d]
  \label{}
\end{align*}
for any finite $n_1,\dotsc,n_d$, it is nevertheless true that equation \eqref{psi_n_expansion} holds for all $x\in\Omega$ if and only if the expression in the square brackets is identically zero.
Furthermore, since $t$ is arbitrary in this expression, this implies that $W_{n_1,\dotsc,n_d}(t,x)$ does not depend on $x$.
After setting $W_{n_1,\dotsc,n_d}(t,x)=V_{n_1,\dotsc,n_d}(t)$, the resulting differential equation has the solution
\begin{align*}
    c_{n_1,\dotsc,n_d,k_1,\dotsc,k_d}(t)=c_{n_1,\dotsc,n_d,k_1,\dotsc,k_d}(t')\exp{\biggl[-\frac{i(t-t')}{2m\hbar}\sum_{j=1}^d\frac{k_j^2 p_j^2}{n_j^2}-\frac{i}{\hbar}\int_{t'}^t V_{n_1,\dotsc,n_d}(s)ds\biggr]},
  \label{}
\end{align*}
which leads to
\begin{align}
  \psi_{n_1,\dotsc,n_d}(t,x)=\exp{\biggl[-\frac{i}{\hbar}\int_{t'}^t V_{n_1,\dotsc,n_d}(s)ds\biggr]}\sum_{k_1=-n_1}^{n_1}\dotsb\sum_{k_d=-n_d}^{n_d} c_{n_1,\dotsc,n_d,k_1,\dotsc,k_d}(t')\nn\\ \times\exp{\biggl[-\frac{i(t-t')}{2m\hbar}\sum_{j=1}^d\frac{k_j^2 p_j^2}{n_j^2}+\frac{i}{\hbar}\sum_{j=1}^d\frac{k_j p_j x_j}{n_j}\biggr]}.
  \label{psi_n_sol}
\end{align}
Setting $t=t'$ in \eqref{psi_n_sol}, multiplying the result by $\exp{\bigl[-\sum_{j=1}^d il_j p_j x_j/(n_j\hbar)\bigr]}$, where $-n_j\le l_j\le n_j$, $1\le j\le d$, and integrating over $x\in\Omega$, we find
\begin{align*}
    c_{n_1,\dotsc,n_d,l_1,\dotsc,l_d}(t')=\frac{\prod_{j=1}^d p_j}{(2\pi \hbar)^d\prod_{j=1}^d n_j}\int_\Omega\psi_{n_1,\dotsc,n_d}(t',x)\exp{\biggl[-\frac{i}{\hbar}\sum_{j=1}^d\frac{l_j p_j x_j}{n_j}\biggr]}dx,
  \label{}
\end{align*}
which leads to
\begin{align}
  \psi_{n_1,\dotsc,n_d}(t,x)=\frac{\prod_{j=1}^d p_j}{(2\pi \hbar)^d\prod_{j=1}^d n_j}\exp{\left[-\frac{i}{\hbar}\int_{t'}^t V_{n_1,\dotsc,n_d}(s)ds\right]}\int_\Omega\psi_{n_1,\dotsc,n_d}(t',x')\nn\\ \times\sum_{k_1=-n_1}^{n_1}\dotsb\sum_{k_d=-n_d}^{n_d}\exp{\biggl[-\frac{i(t-t')}{2m\hbar}\sum_{j=1}^d\frac{k_j^2 p_j^2}{n_j^2}+\frac{i}{\hbar}\sum_{j=1}^d\frac{k_j p_j(x_j-x'_j)}{n_j}\biggr]} dx'.
  \label{psi_n_solution}
\end{align}

Setting
\begin{align*}
  &\psi(t,x)=\lim_{n_1\to\infty}\dotsb\lim_{n_d\to\infty}\psi_{n_1,\dotsc,n_d}(t,x),\\
  &V(t)=\lim_{n_1\to\infty}\dotsb\lim_{n_d\to\infty}V_{n_1,\dotsc,n_d}(t)
\end{align*}
and taking the limit $n_1\to\infty,\dotsc,n_d\to\infty$ in \eqref{psi_n_solution}, we find
\begin{align*}
    \psi(t,x)=\frac{\prod_{j=1}^d p_j}{(2\pi\hbar)^d}\exp{\left[-\frac{i}{\hbar}\int_{t'}^t V(s)ds\right]}\int_{\R^d}\psi(t',x')\nn\\ \times\Biggl\{\prod_{j=1}^d\lim_{n_j\to\infty}\frac{1}{n_j}\int_{\R} \exp{\left[-\frac{i(t-t')k_j^2 p_j^2}{2m\hbar n_j^2}+\frac{ik_j p_j(x_j-x'_j)}{n_j\hbar}\right]dk_j}\Biggr\}dx',
  \label{}
\end{align*}
which finally leads to
\begin{align}
  \psi(t,x)=\left[\frac{m}{2\pi i\hbar(t-t')}\right]^{d/2}\exp{\left[-\frac{i}{\hbar}\int_{t'}^t V(s)ds\right]}\int_{\R^d}\psi(t',x')\exp{\left[\frac{im\norm{x-x'}^2}{2\hbar(t-t')}\right]}dx'.
  \label{psi_solution}
\end{align}

Suppose that, for a fixed $t'$, the function $\psi(t',x)$ is periodic in $x$ with the period
\begin{align*}
  X=(0,\dotsc,0,X_j,0,\dotsc,0),
  \label{}
\end{align*}
where $X_j\not=0$ is in the $j$th position for some $1\le j\le d$, so that
\begin{align} \begin{split}
  \psi(t',x+X)=\psi(t',x).
  \label{psi_periodicity}
\end{split}  \end{align}
It is easy to prove that, at any other time $t$, the function $\psi(t,x)$ is also periodic in $x$ with the same period $X$.
Indeed, from \eqref{psi_solution},
\begin{align*}
    \psi(t,x+X)=\left[\frac{m}{2\pi i\hbar(t-t')}\right]^{d/2}\exp{\left[-\frac{i}{\hbar}\int_{t'}^t V(s)ds\right]}\nn\\ \times\int_{\R^d}\psi(t',x')\exp{\left[\frac{im\norm{x+X-x'}^2}{2\hbar(t-t')}\right]}dx',
  \label{}
\end{align*}
which, after the change of variable $x'\mapsto x'+X$ and use of  \eqref{psi_periodicity}, gives
\begin{align*}
  \psi(t,x+X)=\psi(t,x).
  \label{}
\end{align*}
Applying this result to the case $X_j<2\pi\hbar/p_j$ for any $1\le j\le d$, we obtain the statement.
\end{proof}

\section{Conclusions}

We have studied the evolution of superoscillating initial data for the quantum driven harmonic oscillator.
Our main result shows that superoscillations are amplified by the harmonic potential and that the analytic solution develops a singularity in finite time.
Moreover, even for non-superoscillatory initial data, the harmonic oscillator displays a superoscillatory behavior since the solution contains the term that increases arbitrarily with time (up to the time of singularity).
This phenomenon does not occur for the free particle.

Since the leading semiclassical behavior of the propagator for any quantum-mechanical system is given by the propagator for the driven harmonic oscillator, our results directly apply to a much broader problem of establishing presence of superoscillations in any quantum system (at least in the semiclassical limit).

We have also shown that for a large class of solutions of the Schr\"odinger equation, superoscillating behavior at any given time implies superoscillating behavior at any other time.

\end{document}